\documentclass{llncs}%

\usepackage{booktabs} 
\usepackage{amsmath,amssymb} 
\usepackage{todonotes}
\usepackage{amsfonts}
\usepackage{algorithm}
\usepackage[noend]{algpseudocode}

\newcommand{\monop}{\otimes}
\newcommand{\1}{\mathbf{1}}

\floatstyle{plain}
\newfloat{myalgo}{tbhp}{mya}

\newenvironment{Algorithm}[2][tbh]%
{\begin{myalgo}[#1]
		\centering
		\begin{minipage}{#2}
			\begin{algorithm}[H]}%
			{\end{algorithm}
		\end{minipage}
\end{myalgo}}

\def\odiv{{ \ominus\hspace{-8pt}:}\;}

\def\smallodiv{{ \ominus\hspace{-7.45pt}:}\;}

\begin{document}

\title{Residuation for Soft Constraints: Lexicographic Orders and
        Approximation Techniques\thanks{Research partially supported by the MIUR PRIN 2017FTXR7S ``IT-MaTTerS''
		and by GNCS-INdAM (``Gruppo Nazionale per il Calcolo Scientifico'').}}

\author{Fabio Gadducci\inst{1}, Francesco Santini\inst{2}}
\institute{Dipartimento di Informatica, Universit{\`a}
	di Pisa, Pisa, Italy\\
\email{fabio.gadducci@unipi.it}
\and
Dipartimento di Matematica e Informatica, Universit{\`a}
 di Perugia, Perugia, Italy\\
\email{francesco.santini@unipg.it}\\
}

\maketitle

\begin{abstract}
Residuation theory concerns the study of partially ordered algebraic structures, most often monoids,
equipped with a weak inverse for the monoidal operator.
One of its area of application has been constraint programming, whose 
key requirement is the presence of an aggregator operator for combining preferences.
Given a residuated monoid of preferences, the paper first shows how to build a new residuated monoid 
of (possibly infinite) tuples, which is based on the lexicographic order. Second, it introduces a variant of 
an approximation technique (known as Mini-bucket) that exploits the presence of the weak inverse.
%
\end{abstract}

\section{Introduction}\label{sec:intro}
\emph{Residuation theory}~\cite{residuation} concerns the study of partially ordered algebraic structures, most often just monoids, 
equipped with an operator that behaves as a weak inverse to the monoidal one, without the structure being necessarily a group.
Such structures have since long be investigated in mathematics and computer science. Concerning e.g. logics, residuated monoids form the basis for the semantics of substructural logics~\cite{onoETC}.
As for e.g. discrete event systems such as weighted automata, the use of tropical semirings put forward the adoption of residuals for 
the approximated solution of inequalities~\cite{resbook}. 

One of the recent area of application of residuation theory has been constraint programming. 
Roughly, a  \emph{Soft Constraint Satisfaction Problem} 
is given by a relation on a set of variables, plus a preference 
score to each assignment of such variables~\cite{jacm97,schiex}. 
They key requirement is the presence of an aggregator operator for combining preferences, making such a set a monoid, and a large body 
of work has been devoted to enrich such a structure, guaranteeing that resolution techniques can be generalised by 
a parametric formalism for designing metrics and algorithms.
An example are \emph{local-consistency} algorithms~\cite{arconsistency}, devised for safely moving costs  towards constraints involving a smaller number of variables, 
without changing the set of solutions and their preference.  In order to ``move'' quantities, we need to  ``subtract'' costs somewhere and ``add'' them elsewhere. 

The paper focuses on residuated monoids for constraint programming. Their relevance for local-consistency, as mentioned above, has been spotted 
early on~\cite{residuation1,resCS}, and various extensions  has been proposed~\cite{ipl}, as well as applications to languages based on the Linda paradigm, such as 
\emph{Soft Concurrent Constraint Programming},  where a process may be \emph{telling} and \emph{asking} constraints to a centralised store~\cite{labelled}. 
More precisely, we tackle here two aspects. On the one side, we consider lexicographic orders, as used in contexts with multi-objective
problems. That is, the preference values are obtained by the combination of separate concerns, and the order of the combination matters. 
On the other side, we introduce a soft version of \emph{Bucket} and \emph{Mini-bucket} elimination algorithms, well-known exact and approximated techniques for inference,
which exploits the presence of a residuated monoid: in order to have an estimation of  the approximation on the preference level of a solution, it is necessary to use a removal operator. Finally we present a \emph{Depth-First Branch-and-Bound} algorithm, which exploits upper and lower bounds to prune search. 
Our proposals generalise the original soft versions of these approximation techniques presented in~\cite{bucketsemiring}.

Lexicographic orders are potentially useful in applications that involve multiple objectives and attributes,
and as such have been extensively investigated in the literature on soft constraints.
However, usually the connection has been established by encoding a lexicographic \emph{hard} constraint problem, where the preference structure is a Boolean algebra, 
into a soft constraint formalism. 
For example, in \cite{freuderlexi} the authors show how to encode a lexicographic order and how the resulting structure
can support specialised algorithms such as \emph{Branch-and-bound}. 
%
\emph{Hierarchical Constraint Logic Programming}~\cite{hierarchicalclp} frameworks allow to handle 
both hard constraints  and several preference levels of soft constraints, whose violations need to be minimised, and such levels are usually managed 
following a lexicographic order~\cite{valuation}. 
%
However, even if lifting the algebraic structure of a preference set to the associated set of (possibly infinite) tuples with a point-wise order is straightforward, 
doing the same for the lexicographic order is not, and this result cannot be directly achieved for the formalisms in \cite{jacm97,schiex}. 
The solution advanced in~\cite{GadducciHMW13,valuation} is to drop some preference values from the domain carrier of the set of tuples. 
The present work builds on
this proposal by dealing with sets of preferences that form residuated monoids, systematise and extending the case of infinite tuples tackled in~\cite{sca} 
to tuples of any length.

The paper has the following structure: in Section~\ref{sec:om} we present  the background on partially ordered residuated monoids, which is the structure 
we adopt to model preferences. 
In Section~\ref{sec:collapsing} we consider the collapsing elements of a monoid, which will be used to define an ad-hoc algebraic structure representing 
(possibly infinite) lexicographically ordered tuples of elements of the chosen monoid, which is given in Section~\ref{sec:lexico}. 
The latter section also presents our main construction, introducing residuation for these lexicographically ordered monoids.  
Section~\ref{sec:bucket} shows how residuation helps to find a measure of goodness  between an algorithm and its tractable approximation. 
Finally, in Section~\ref{sec:conclusion} we wrap up the paper with concluding remarks and ideas about future works.

\section{Preliminaries}\label{sec:om}


This section recalls some of the basic algebraic structures needed for defining
the set of preference values. In particular, we propose elements of \emph{ordered monoids} 
to serve as preferences, which allows us to compare and compose preference values.

\subsection{Ordered Monoids}
\label{sec:lem}

The first step is to define an algebraic structure for modelling preferences. 
We refer to~\cite{ipl} for the missing proofs as well as for an introduction and a comparison with other proposals.

\begin{definition}[Orders]
	A partial order (PO) is a pair $\langle A, \leq \rangle$ such that
	$A$ is a set 
	and $\leq \,\,\subseteq A \times A$ is a reflexive, transitive, and
	anti-symmetric  relation.
	%
	%
	A join semi-lattice (simply semi-lattice, SL) is a POs
	such that any finite subset of $A$ has a least upper bound (LUB);
	a complete lattice (CL) is a PO such that any subset of A has a LUB.
\end{definition}

The LUB of a subset $X \subseteq A$ is denoted $\bigvee X$, and it is unique. 
Note that we require the existence of $\bigvee \emptyset$, which is the bottom
of the order, denoted as $\bot$, and sometimes we will talk about a PO with bottom element
(POB).
The existence of LUBs for any subset of $A$ (thus including $\emptyset$) guarantees that 
CLs also have greatest lower bounds (GLBs) for any subset $X$ of $A$: 
it will be denoted by $\bigwedge X$. 
Whenever it exists, 
$\bigvee A$ corresponds to the top of the order, denoted as $\top$.
%

%


%
%

\begin{definition}[Ordered monoids]\label{defn:clm}
	A (commutative) monoid is a triple
	$\langle A, \monop,$ $\1 \rangle$ such that $\monop: A \times A \rightarrow A$ is
	a commutative and associative function and $\1 \in A$ is its \emph{identity} element,
	i.e., $\forall a \in A. a \monop \1 = a$.
	
	A partially ordered monoid (POM) is a 4-tuple
	$\langle A, \leq, \monop, \1 \rangle$ such that 	
	$\langle A, \leq \rangle$ is a PO and $\langle A, \monop, \1 \rangle$ a monoid.
	A semi-lattice monoid (SLM) and a complete lattice monoid (CLM) are 
	POMs such that their underlying PO is a SL, a CL respectively.
\end{definition}

For ease of notation, we use the infix notation: $a \monop b$ stands for $\monop(a,b)$.

\begin{example}[Power set]\label{ex:powerset}
	Given a (possibly infinite) set $V$ of variables, we consider
	the monoid $\langle 2^V, \cup, \emptyset \rangle$
	of (possibly empty) subsets of $V$, with union as the monoidal operator.
	Since the operator is idempotent (i.e., $\forall a\in A.\, a \monop a = a$), 
	the natural order ($\forall a, b \in A.\, a \leq b$ iff $a \monop b = b$) 
	is a partial order, and 
	it coincides with subset inclusion:
	in fact, $\langle 2^V, \subseteq, \cup, \emptyset \rangle$
	is a CLM.
\end{example}

In general, the partial order $\leq$ and the multiplication $\otimes$ can be unrelated.
This is not the case for distributive CLMs.

\begin{definition}[Distributivity]
	\label{dist}
	A SLM $\langle A, \leq, \monop, \1 \rangle$ is finitely distributive if
	\[ \forall X \subseteq_f A.\, \forall a \in A.\, \quad a \monop  \bigvee X = \bigvee \{a \monop x \mid x \in X\}. \]
	A CLM is distributive is the equality holds also for any subset.
\end{definition}

In the following, we will sometimes write $a \otimes X$ for the set  $\{a \monop x \mid x \in X\}$.

\begin{remark}
Note that $a \leq b$ is equivalent to $\bigvee \{a,b\} = b$ for all $a, b \in A$.
Hence, finite distributivity implies that $\otimes$ is monotone with respect to $\leq$ 
(i.e., $\forall a, b, c \in A.\, a \leq b \Rightarrow a \otimes c \leq  b \otimes c$)
and that $\bot$ is the zero element
of the monoid (i.e., $\forall a \in A.\, a\otimes  \bot = \bot$).
The power-set CLM in Example~\ref{ex:powerset} is distributive.
\end{remark}

\begin{example}[Extended integers]\label{ex:bipolar}
	The extended integers $\langle \mathbb{Z} \cup \{\pm \infty\}, \leq, +, 0 \rangle$, where 
	$\leq$ is the natural order, such that for $k \in \mathbb{Z}$
	$$-\infty \leq k \leq +\infty,$$
	$+$ is the natural addition, such that for $k \in \mathbb{Z} \cup \{+\infty\}$
	$$\pm\infty + k = \pm \infty, \qquad \qquad +\infty + (-\infty) = -\infty,$$
	and $0$ the identity element constitutes a distributive CLM,
	and $+\infty$ and $-\infty$ are respectively the top and the bottom 
	element of the CL.
\end{example}

%
%

\begin{remark}
	Finitely distributive SLMs precisely corresponds to \emph{tropical} semirings 
	by defining the (idempotent) sum operator as
	$a \oplus b = \bigvee \{a, b\}$ for all $a,b \in A$.
	If, moreover, $\1$ is the top of the SLM we end up 
	with \emph{absorptive} semirings~\cite{golan}, 
	which are known as $c$-semirings 
	in the soft constraint jargon~\cite{jacm97}.
	Together with monotonicity, imposing $\1$ to coincide with $\top$ means 
	that preferences are negative (i.e., $a \leq \1$ for all $a \in A$).
	
	Distributive CLMs are known in the literature as \emph{quantales}~\cite{quantales}.
\end{remark}

\begin{remark}
	\label{remark}
	Given two distributive CLMs, it is easy to show that their Cartesian product, whose elements are pairs and where the partial order and the monoidal operator are 
	defined point-wise, is a distributive CLM.
	In particular, in the following we consider the Cartesian product of
	$\langle \mathbb{Z} \cup \{\pm \infty\}, \leq, +, 0 \rangle$ with itself: its set of
	elements is $(\mathbb{Z} \cup \{\pm \infty\})^2$, the identity element is $(0,0)$,
	and the top and bottom elements are $(+\infty, +\infty)$ and $(-\infty, -\infty)$, respectively.
\end{remark}

\subsection{Residuated monoids}

We first introduce \emph{residuation}, which allows us to define a ``weak'' inverse 
operator with respect to the monoidal operator $\otimes$. 
In this way, besides aggregating values together, it is also possible to remove one from 
another. Residuation theory~\cite{golan} is concerned with the study of sub-solutions of the 
equation $b \otimes x = a$, where $x$ is a ``divisor'' of $a$ with respect to $b$. 
The set of sub-solutions of an equation contains also the 
possible solutions, whenever they exist, and in that case the maximal element is also a 
solution.

\begin{definition}[residuation]
	A residuated POM is a 5-tuple $\langle A,$ $\leq, \otimes,  \odiv, \1 \rangle$ such that
	$\langle A, \leq \rangle$ is a PO,  $\langle A, \otimes, \1 \rangle$ is a monoid, and
	$\odiv: A \times A \rightarrow A$ is a function such that 
	\begin{itemize}
		\item $\forall a, b, c \in A.\ b \otimes c \leq a \iff c \leq a \odiv b$.
	\end{itemize}
\end{definition}

In the following, we will sometimes write $a \odiv X$ and $X \odiv a$
for the set  $\{a \odiv x \mid x \in X\}$ and  $\{x \odiv a \mid x \in X\}$,
respectively.

\begin{remark}
It is easy to show that residuation is monotone on the first argument and
anti-monotone on the second. In fact, in a SML 
$\bigvee (X \odiv a) \leq \bigvee X \odiv a$,
and the same in a CLM with respect to infinite sub-sets.
However, the equality does not hold,
e.g. in the Cartesian product of the CLM  
$\langle N \cup \{\infty\}, \geq, +, 0 \rangle$ with itself.

Also, 
$a \odiv \bigvee X \leq \bigwedge (a \odiv X)$
whenever the latter exists, as it does in CLMs.
\end{remark}

\begin{remark}
\label{remarkC}
As for distributivity, given two residuated POMs, it is easy to show that their Cartesian product
is a residuated POM.
\end{remark}

Residuation implies distributivity (see e.g.~\cite[Lem. 2.2]{ipl}).

\begin{lemma}
Let $\langle A, \leq, \monop, \1 \rangle$ be a residuated POM. 
Then it is monotone.
If additionally it is a SLM (CLM), then it is finitely distributive
(distributive).
\end{lemma}

Conversely, it is noteworthy that CLMs are always residuated,
and the following folklore fact holds.
%

\begin{lemma}
\label{resCLM}
Let $\langle A, \leq, \monop, \1 \rangle$ be a distributive CLM. 
It is residuated and 
$\forall a, b \in A.\ a \odiv b = \bigvee \{c \mid b \otimes c \leq a\}$.
\end{lemma}

We close with a simple lemma relating residuation with the top and the bottom 
elements of a POM.

\begin{lemma}
	\label{someProps}
	Let $\langle A, \leq, \monop, \odiv, \1 \rangle$ be a residuated POM. 
	If it has the bottom element $\bot$, then it also has the top element $\top$and
	$\forall a \in A.\ a \odiv \bot = \top$.
	Viceversa, if it has the top element $\top$ then
	$\forall b \in A.\ \top \odiv b = \top$.
\end{lemma}

\begin{remark}
	Nothing can be stated for $\bot \odiv a$, since there could be elements that are 
	$\bot$-divisors: see again the Cartesian product of the CLM  
	$\langle \mathbb N \cup \{\infty\}, \geq, +, 0 \rangle$ with itself, where 
	$\langle \infty, 3\rangle  \otimes \langle 4, \infty \rangle = \langle \infty,\infty \rangle$. 
	
	Similarly, nothing can be stated for $a \odiv \top$: see the Cartesian 
	product of the CLM  
	$\langle \mathbb N \cup \{\infty\}, \geq, +, 0 \rangle$ with its dual CLM
	$\langle \mathbb N \cup \{\infty\}, \leq, +, 0 \rangle$.
\end{remark}


\section{The ideal of collapsing elements}
\label{sec:collapsing}
As shown in~\cite{GadducciHMW13}, the first step for obtaining SLMs based on a lexicographic order 
is to restrict the carrier of the monoid. 

\begin{definition}
	Let $\langle A, \monop, \1 \rangle$ be a monoid. Its sub-set $I(A)$ 
	of \emph{cancellative} elements is defined as 
	$\{ c \mid \forall a, b \in A.\ a \otimes c = b \otimes c \implies a = b \}$.
\end{definition}

We  recall a well-known fact.

\begin{lemma}
	\label{ideal}
	Let $\langle A, \monop, \1 \rangle$ be a monoid.
	Then $I(A)$ is a sub-monoid of $A$ and $C(A) = A\setminus I(A)$ 
	is a prime ideal of $A$.
\end{lemma}

Explicitly, $C(A) = \{ c \mid \exists a, b \in A.\ a \neq b \wedge a \otimes c = b \otimes c\}$.
Being an ideal means that $\forall a \in A, c \in C( A).\ a \otimes c \in C(A)$,
and being prime further states that 
$\forall a, b \in A.\ a\otimes b \in C( A) \implies a \in C(A) \vee b \in C(A)$.
All the proofs are straightforward, and we denote $C(A)$ as the set of \emph{collapsing} 
elements of $A$. 

Note that an analogous closure property does not hold for LUBs.

\begin{example}\label{flat}
	Consider the monoid of natural numbers $\langle \mathbb N, +, 0 \rangle$ and 
	the (non distributive) CLM with elements $\mathbb N \cup \{\bot,\top\}$ obtained by lifting 
	the flat order (i.e., $a \not \leq b$ for any $a, b \in \mathbb N$ as well as 
	$a + \bot = \bot = \top + \bot$ and $a + \top = \top$  for any $a \in \mathbb N$).
	Then, $I(\mathbb N \cup \{\bot,\top\}) = \mathbb N$ is not closed under finite LUBs. 
	
	Now, let us consider the distributive CLM with elements $\mathbb N \cup \{\infty\}$ 
	obtained by lifting the natural order induced by addition.
	We have that $I(\mathbb N \cup  \{\infty\}) = \mathbb N$ is a (finitely distributive) SLM, yet
	it is not closed with respect to infinite LUBs.
\end{example}

We now present a simple fact that is needed later on.

\begin{lemma}
	Let $A_1, A_2$ be POMs and $A_1 \times A_2$ their Cartesian product.
	Then we have $C(A_1 \times A_2) = C(A_1) \times A_2 \cup A_2 \times C(A_2)$. 
\end{lemma}

\begin{example}
	Let us consider the tropical SLM $\langle \mathbb N \cup \{\infty\}, \geq, +, 0 \rangle$ 
	and the Cartesian product with itself. 
	Clearly, $C(\mathbb N\times \mathbb N)$ is not closed under finite LUBs: it suffices to consider 
	$X= \{\langle \infty, 3 \rangle, \langle 4, \infty \rangle\} \subseteq C(\mathbb N\times \mathbb N)$,
	since $\bigvee X  = \langle 3,4 \rangle \not \in C(\mathbb N\times \mathbb N)$.
	Neither is $C(\mathbb N\times \mathbb N)$ closed under residuation, as suggested by 
	Lem.~\ref{someProps}, since the top element is not necessarily collapsing.
	Indeed, in $C(\mathbb N\times \mathbb N)$ we have 
	$\langle \infty, 4 \rangle  \odiv \langle \infty, 3 \rangle = \langle 0,1\rangle$.
\end{example}

\begin{remark}
	Note that in an absorptive CLM $A$ we have that $a \odiv b = 1$ 
	whenever $b \leq a$. Hence $C(A)$ is usually not closed under residuation, since 
	$1$ is cancellative.
\end{remark}

\subsection{A different view on collapsing elements}

When the first presentation of lexicographic 
SLMs was provided~\cite{GadducciHMW13}, a different 
set of collapsing elements was considered.

\begin{definition}[\cite{GadducciHMW13}]
	Let $\langle A, \leq, \monop, \1 \rangle$ be a POM. Its sub-set $C'(A)$ 
	is defined as 
	$\{ c \mid \exists a, b \in A.\ a < b \wedge a \otimes c = b \otimes c\}$.
\end{definition}

Clearly, $C'(A) \subseteq C(A)$. However, we can replicate Lem.~\ref{ideal}.
%

\begin{lemma}
	Let $\langle A, \monop, \1 \rangle$ be a monoid.
	Then $C'(A)$ is an ideal of $A$.
	If $\monop$ is monotone, then 
	$I'(A) = A \setminus C'(A) $ is a sub-monoid of $A$
	and $C'(A)$ a prime ideal of $A$.
\end{lemma}

Explicitly, $I'(A) = \{ c \mid \forall a, b \in A.\ a \otimes c = b \otimes c \implies a \not < b \}$.
The definitions we encounter in the next section 
could then be rephrased using 
$I'(A)$ and $C'(A)$ with minimal adjustments, thus confirming the proposal
in~\cite{GadducciHMW13}.\footnote{And in fact, the lemma holds also for 
a property that is weaker than monotonicity: it suffices that 
$\forall a, b, c.\ a \leq b \implies (a \otimes c \leq b \otimes c) \vee (b \otimes c \leq a \otimes c)$.
}

However, what is in fact noteworthy is that the two approaches are coincident
whenever distributivity holds, as shown by the lemma below.

\begin{lemma}
	Let $\langle A, \leq, \monop, \1 \rangle$ be a finitely distributive SLM.
	Then $C'(A) = C(A)$.
\end{lemma}
\begin{proof}
	We already noted that $C'(A) \subseteq C(A)$ always holds. Now, let $a, b, c$ such that
	$a \neq b \wedge a \otimes c = b \otimes c$: it suffices to consider $a \vee b$, noting 
	that it must be either $a < a \vee b$ or $b < a \vee b$ and that by distributivity 
	$(a \vee b) \otimes c = a \otimes c = b \otimes c$.
\end{proof}

\begin{remark}
Consider the (non distributive) CLM $\langle \mathbb [0 \ldots n] \cup \{\bot,\top\}, +, 0 \rangle$ 
obtained by lifting the initial segment $[0 \ldots n]$ of the natural numbers with the flat order (as done for
the CLM of all natural numbers in Example~\ref{flat}). Here addition is capped, so that 
e.g. $n + m = n$ for all $m$. Hence, $C([0 \ldots n] \cup \{\bot,\top\}) = [1 \ldots n] \cup \{\bot,\top\}$
that is, all elements except $0$. Instead, $C'([0 \ldots n] \cup \{\bot,\top\})) =  \{\bot,\top\}$.
\end{remark}

\section{On lexicographic orders}\label{sec:lexico}
We now move to lexicographic orders, taking into account the results in Sect.~\ref{sec:collapsing}.

\begin{proposition}\label{def:lexilist}
	Let $\langle A, \leq, \monop, \1 \rangle$ be a POM with bottom element $\bot$.
	Then we can define a family $\langle Lex_k(A), \leq_k, \monop^k, \1^k \rangle$ 
	of POMs with bottom element $\bot^k$ such that $\monop^k$ is defined point-wise, 
	$Lex_1(A) = A$ and $\leq_1 = \leq$, and
	
	\begin{itemize}
		\item $Lex_{k+1}(A) = I(A) Lex_k(A) \cup C(A) \{\bot\}^k$,
		\item $a_1 \ldots a_k \leq_k b_1 \ldots b_k$ if $a_1 < b_1$ or $a_1 = b_1$ 
		and  $a_2 \ldots a_k \leq_{k-1} b_2 \ldots b_k$.
	\end{itemize}
\end{proposition}

\begin{proof}
	Monoidality of $\monop^k$ as well as reflexivity and symmetry of $\leq_k$ are straightforward.
	As for transitivity, let $a_1 \ldots a_k \leq_k b_1 \ldots b_k$ and $b_1 \ldots b_k \leq_k c_1 \ldots c_k$.
	If $a_1 \leq b_1 < c_1$ or $a_1 < b_1 \leq c_1$, it follows immediately; if $a_1 = b_1 = c_1$, then it
	holds by induction.
\end{proof}

Note that $Lex_k(A)$ is contained in the $k$-times Cartesian product $A^k$, 
and the definitions of $\monop^k$, $\1^k$, and $\bot^k$ coincide.
Also, the bottom element is needed for padding the tuples, in order to make
simpler the definition of the order.

We can provide an alternative definition for such POMs.

\begin{lemma}
	Let $\langle A, \leq, \monop, \1 \rangle$ be a POM with bottom element $\bot$.
	Then $Lex_{k+1}(A) = \bigcup_{i \leq k}I(A)^i A \{\bot\}^{k-i}$ for all $k$.
\end{lemma}
\begin{proof}
	The proof goes by induction on $k$. For $k = 0$ it is obvious.
	Let assume it to be true by induction for $k = n$, that is,
	$Lex_{n+1}(A) = \bigcup_{i \leq n}I(A)^i A \{\bot\}^{n-i}$ .
	Now, 
	$Lex_{n+2}(A) =  I(A) Lex_{n+1}(A) \cup C(A) \{\bot\}^{n+1} =
	\bigcup_{i \leq n}I(A)^{i+1} A \{\bot\}^{n-i} \cup C(A) \{\bot\}^{n+1} =
	\bigcup_{i \leq n+1}I(A)^{i} A \{\bot\}^{n+1-i}$
	and we are done, where the latter equality holds since
	$I(A)\{\bot\}^{n+1} \in \bigcup_{i \leq n}I(A)^{i+1} A \{\bot\}^{n-i}$
	and $A = I(A) \cup C(A)$.
\end{proof}

Now, given a tuple $a$ of elements in $A^k$, for $i \leq k$ we denote with $a_i$ its 
$i$-th component and with $a_{\mid i}$ its prefix $a_1 \ldots a_i$, with the obvious 
generalisation for a set $X \subseteq A^k$, noting that $a_1 = a_{\mid 1}$.

\begin{theorem}\label{theo:lexiSLM}
	Let $\langle A, \leq, \monop, \1 \rangle$ be a finitely distributive SLM (distributive CLM).
	Then so is $\langle Lex_k(A), \leq_k, \monop^k, \1^k \rangle$ for all $k$.
\end{theorem}
\begin{proof}
	As a first step, we define the LUB $\bigvee X$ of a set $X \subseteq Lex_k(A)$,
	given inductively as
	\begin{itemize}
		\item $(\bigvee X)_1 = \bigvee X_1 = \bigvee \{ y \mid y \in X_{\mid 1}\}$
		\item $(\bigvee X)_{i+1} = \bigvee \{ y \mid (\bigvee X)_1 \ldots (\bigvee X)_i y \in X_{\mid i+1}\}$
	\end{itemize}
	
	Now, $\bigvee X$ is clearly a suitable candidate, since $x \leq_k \bigvee X$ for all 
	$x \in X$. Minimality is proved inductively by exploiting the analogous 
	definition of $\leq_k$.
	
	Concerning distributivity, we need to show that
	$(\bigvee a \otimes^k X)_n = a_n \otimes (\bigvee X)_n$ holds for all $n \leq k$.
	We proceed by induction on $n$. If $n = 1$, this boils down to the distributivity of 
	the underlying monoid, since
	$$(\bigvee a \otimes^k X)_1 = \bigvee (a \otimes^k X)_1 = \bigvee \{ y \mid y \in (a \otimes^k X)_{\mid 1}\} =$$
         $$\bigvee \{ y \mid \exists z \in X_{\mid 1}.\ y = a_1 \otimes z\} = \bigvee \{ a_1 \otimes z \mid z \in X_{\mid 1}\}$$
	and 
	$$ a_1 \otimes (\bigvee X)_1 = a_1 \otimes \bigvee \{ y \mid y \in X_{\mid 1}\}  = \bigvee \{ a_1 \otimes y \mid y \in X_{\mid 1}\}$$
	
	Now, let us assume that it holds for $n$. Now we have 
	$$(\bigvee a \otimes^k X)_{n+1} = \bigvee \{ y \mid (\bigvee a \otimes^k X)_1 \ldots (\bigvee a \otimes^k X)_n y \in (a \otimes^k X)_{\mid n+1}\} =$$
	$$ \bigvee \{ y \mid \exists z \in X_{n+1}.\ y = a_{n+1} \otimes z \wedge a_1 \otimes (\bigvee X)_1 \ldots a_n \otimes (\bigvee X)_n y \in (a \otimes^k X)_{\mid n+1}\} =$$
	$$ \bigvee \{ a_{n+1} \otimes z \mid (\bigvee X)_1 \ldots (\bigvee X)_n z \in X_{\mid n+1}\}$$
	The latter equality obviously holds if $a_1 \ldots a_n$ are cancellative, yet it holds also otherwise since in that case 
	$a_{n+1} = \bot$, hence both sides coincide with $\bot$. Finally
	$$a_{n+1} \otimes (\bigvee X)_{n+1} = a_{n+1} \otimes \bigvee \{ y \mid (\bigvee X)_1 \ldots (\bigvee X)_n y \in X_{\mid n+1}\}  =$$
	 $$ \bigvee \{ a_{n+1} \otimes y \mid (\bigvee X)_1 \ldots (\bigvee X)_n y \in X_{\mid n+1}\}$$
\end{proof}

\subsection{On lexicographic residuation}
The fact that $Lex_k(A)$ is a CLM if so is $A$ tells us that $Lex_k(A)$ is
also residuated. 

\begin{example}
	Let us consider the usual tropical CLM of natural numbers with inverse order,
	and the CLM $Lex_{2}(\mathbb N)$. Clearly $C(\mathbb N) = +\infty$. We then have for example
	that
	\[
	(3, 6) \odiv_2 (4, 2) = \bigvee \{(x, y) \mid (4 + x, 2 + y) \leq_2 (3, 6)\} = (0,0)
	\]
	Indeed, $(4 + x, 2 + y) \leq_2 (3, 6)$ holds for any possible choice of $(x, y)$,
	since $4 + x < 3$ for all $x$, hence $(0,0)$ as the result.

	Note that  for the CLM obtained via the Cartesian product $\mathbb N \times \mathbb N$, the 
	result would have been $(0, 4)$.
\end{example}

Indeed, this can be proved in general for POMs. First, we need some additional definitions and technical lemmas.

\begin{definition}
	Let $\langle A, \leq, \monop, \odiv, \1 \rangle$
	be a residuated POM with bottom and $a, b \in Lex_k(A)$. Then 
	\begin{itemize}
		\item $\gamma(a, b) = min\{ i \mid (a_i \odiv b_i) \in C(A)\}$
		\item $\delta(a, b) = min\{ i \mid (a_i \odiv b_i) \otimes b_i < a_i\}$
	\end{itemize}
	with the convention that the result is $k+1$ whenever the set if empty.
\end{definition}


\begin{lemma}\label{limit}
	Let $\langle A, \leq, \monop, \odiv, \1 \rangle$
	be a residuated POM with bottom and $a, b \in Lex_k(A)$. Then 
	either $\delta(a,b) = k+1$ or $\delta(a,b) \leq \gamma(a,b)$.
\end{lemma}
\begin{proof}
	If $\gamma(a,b) < \delta(a,b) \leq k$ then 
	$a_{\gamma(a,b)}  = 
	(a_{\gamma(a,b)} \odiv b_{\gamma(a,b)}) \otimes b_{\gamma(a,b)}$.
	Since $C(A)$ is an ideal of $A$, it holds that $a_{\gamma(a,b)} \in C(A)$, 
	which in turn implies that $a_{\delta(a,b)} = \bot$,
	hence a contradiction.
\end{proof}

We can then present the definition of residuation for lexicographic POMs only
for the cases identified by the proposition above.

\begin{proposition}
	\label{div0}
	Let $\langle A, \leq, \monop, \odiv, \1 \rangle$
	be a residuated POM with bottom and $a, b \in Lex_k(A)$.
	If $\delta(a,b) = \gamma(a,b) = k+1$
	then their residuation $a \odiv_k b$ 
	in $Lex_k(A)$ exists and it is given by 
	\[
	(a_1 \odiv b_1) \ldots  (a_k \odiv b_k) 
	\]
\end{proposition}

\begin{proof}
First of all, note that 
$a \odiv_k b \in I(A)^k \subseteq Lex_k(A)$.
%
So, given $c \in Lex_k(A)$, we need to prove that
$b\otimes^k c \leq_k a$ iff $c \leq_k a \odiv_k b$.
	
\begin{description}
	\item[$\mathbf{[b\otimes^k c \leq_k a].}$]
	Let $l = min \{ i \mid b_i \otimes c_i < a_i\}$,
	with the convention that the result is $k+1$ whenever the set if empty.
	Also, let $m = min\{l, k\} < \delta(a,b)$.

        We have that $b_j \otimes c_j \leq a_j$ for all $j \leq m$, hence
	$c_j \leq a_j \odiv b_j $ for all $j \leq m$. 
	If $c_n < a_n \odiv b_n$ for some $n \leq m$
	we are done. 
	Otherwise, if $m = k < l$ then 
	$c_n = a_n \odiv b_n$ for all $n \leq k$ and we are done.
	Finally, if $m = l \leq k $ then 
	$b_l \otimes c_l = b_l  \otimes (a_l \odiv b_l) = a_l$
	since $l < \delta(a,b)$, hence a contradiction.
		
	\item[$\mathbf{[c \leq_k a \odiv_k b].}$]
	Let $l = min \{ i \mid c_i < (a \odiv_k b)_i\}$,
	with the convention that the result is $k+1$ whenever the set if empty.
	Also, let $m = min \{l, k\} < \delta(a,b)$.
	
        We have that $c_j \leq (a \odiv_k b)_j = a_j \odiv b_j$ for all $j \leq m$, hence
	$b_j \otimes c_j \leq a_j$ for all $j \leq m$.
	If $b_n \otimes c_n < a_n$ for some $n \leq m$
	we are done. 
	Otherwise, if $m = k < l$ then 
	$b_n \otimes c_n = a_n$ for all $n \leq k$	and we are done.
	Finally, if $m = l \leq k$ then 
	$b_l \otimes c_l = a_l = b_l \otimes (a_l \odiv b_l)$
	since $l < \delta(a,b)$, hence either a contradiction
	if $b_l \in I(A)$ or $b_p \otimes c_p = \bot = a_p$ for all $p > l$
	if $b_l \in C(A)$ and consequently $a_l \in C(A)$.
\end{description}
\end{proof}

Note that $a \odiv_k b$ here coincides with the residuation $a \odiv^k b$ 
on the Cartesian product. Furthermore, we have 
that $(a \odiv_k b) \otimes^k b = a$.

\begin{proposition}
	\label{div1}
	Let $\langle A, \leq, \monop, \odiv, \1 \rangle$
	be a residuated POM with bottom and $a, b \in Lex_k(A)$.
	If $\delta(a,b) < \gamma(a,b)$ then their residuation $a \odiv_k b$ 
	in $Lex_k(A)$ exists and it is given by
	\[
	(a_1 \odiv b_1) \ldots  (a_{\delta(a, b)}  \odiv b_{\delta(a, b)}) 
	(\bigvee Lex_{k - \delta(a, b)}(A))
	\]
\end{proposition}

\begin{proof}
First of all, note that 
$(a \odiv_k b)_{\mid \delta(a,b)} \in I(A)^{\delta(a,b)} \subseteq Lex_{\delta(a,b)}(A)$.
%
So, given $c \in Lex_k(A)$, we need to prove that
$b\otimes^k c \leq_k a$ iff $c \leq_k a \odiv_k b$.
	
\begin{description}
	\item[$\mathbf{[b\otimes^k c \leq_k a].}$]
	Let $l = min \{ i \mid b_i \otimes c_i < a_i\}$,
	with the convention that the result is $k+1$ whenever the set if empty.
	Also, let $m = min \{l, \delta(a,b)\}$.
	
        We have that $b_j \otimes c_j \leq a_j$ for all $j \leq m$, hence
	$c_j \leq a_j \odiv b_j $ for all $j \leq m$. 
	If $c_n < a_n \odiv b_n$ for some $n \leq m$
	we are done. 
	Otherwise, if $m = \delta(a,b) \leq l$ then 
	$c_j \leq a_j \odiv b_j $ for all $j \leq \delta(a,b)$ and we are done.
	Finally, if $m = l < \delta(a,b)$ then 
	$b_l \otimes c_l = b_l \otimes (a_l \odiv b_l) = a_l$,
	hence a contradiction.
		
	\item[$\mathbf{[c \leq_k a \odiv_k b].}$]
	Let $l = min \{ i \mid c_i < (a \odiv_k b)_i\}$,
	with the convention that the result is $k+1$ whenever the set if empty.
	Also, let $m = min \{l, \delta(a,b)\}$.
	
        We have that $c_j \leq (a \odiv_k b)_j = a_j \odiv b_j$ for all $j \leq m$, hence
	$b_j \otimes c_j \leq a_j$ and $b_j \otimes c_j \leq b_j \otimes (a_j \odiv b_j)$ 
	for all $j \leq m$, the latter by monotonicity of $\otimes$.
	If $b_n \otimes c_n < a_n$ for some $n \leq m$
	we are done. 
	Otherwise, if $m = \delta(a,b) \leq l$ then 
	$b_{\delta(a,b)} \otimes (a_{\delta(a,b)} \odiv b_{\delta(a,b)}) <
	a_{\delta(a,b)} = b_{\delta(a,b)} \otimes c_{\delta(a,b)} \leq
	b_{\delta(a,b)} \otimes (a_{\delta(a,b)} \odiv b_{\delta(a,b)})$,
	hence a contradiction.
	Finally, if $m = l < \delta(a,b)$ then 
	$b_l \otimes c_l = a_l = b_l \otimes (a_l \odiv b_l)$, 
	hence either a contradiction
	if $b_l \in I(A)$ or $b_p \otimes c_p = \bot = a_p$ for all $p > l$
	if $b_l \in C(A)$ and consequently $a_l \in C(A)$.
\end{description}
\end{proof}

Additionally, please note that $\bigvee Lex_n(A)$ can be easily characterised: 
it coincides with $\top^n$
if $\top \in I(A)$, and with $\top \bot^{n -1}$ otherwise.

\begin{proposition}
	\label{div2}
	Let $\langle A, \leq, \monop, \odiv, \1 \rangle$
	be a residuated POM with bottom element $\bot$ and $a, b \in Lex_k(A)$.
	If $\delta(a,b) = \gamma(a,b) \leq k$ or $\delta(a,b) >  \gamma(a,b)$
	then their residuation $a \odiv_k b$ 
	in $Lex_k(A)$ exists and it is given by
	\[
	(a_1 \odiv b_1) \ldots  (a_{\gamma(a, b)}  \odiv b_{\gamma(a, b)}) 
	\bot^{k - \gamma(a, b)}
	\]
\end{proposition}

\begin{proof}
First of all, note that 
$(a \odiv_k b)_{\mid \gamma(a,b)} \in I(A)^{\gamma(a,b)-1} C(A) \subseteq Lex_{\gamma(a,b)}(A)$.
Also, $\delta(a,b) >  \gamma(a,b)$ implies that 
$\delta(a,b) = k+1$ and $a_{\gamma(a,b)} \in C(A)$.
%
Given $c \in Lex_k(A)$, we need to prove that
$b\otimes^k c \leq_k a$ iff $c \leq_k a \odiv_k b$.
	
\begin{description}
	\item[$\mathbf{[b\otimes^k c \leq_k a].}$]
	Let $l = min \{ i \mid b_i \otimes c_i < a_i\}$,
	with the convention that the result is $k+1$ whenever the set if empty.
	Also, let $m = min \{l, \gamma(a,b)\} \leq \delta(a,b)$.

        We have that $b_j \otimes c_j \leq a_j$ for all $j \leq m$, hence
	$c_j \leq a_j \odiv b_j $ for all $j \leq m$. 
	If $c_n < a_n \odiv b_n$ for some $n \leq m$
	we are done. 
	Otherwise, if $m = \gamma(a,b) \leq l$ then 
	$c_{\gamma(a,b)} \in C(A)$ and we are done since
	$c_p = \bot = (a \odiv_k b)_p$ for all $p > {\gamma(a,b)}$.
	Finally, if $m = l < \gamma(a,b)$ then 
	$b_l \otimes c_l = b_l \otimes (a_l \odiv b_l) = a_l$
	since $l < \delta(a,b)$, hence a contradiction.
		
	\item[$\mathbf{[c \leq_k a \odiv_k b].}$]
	Let $l = min \{ i \mid c_i < (a \odiv_k b)_i\}$,
	with the convention that the result is $k+1$ whenever the set if empty.
	Also, let $m = min \{l, \gamma(a,b)\} \leq \delta(a,b)$.
	
        We have that $c_j \leq (a \odiv_k b)_j = a_j \odiv b_j$ for all $j \leq m$, hence
	$b_j \otimes c_j \leq a_j$ and $b_j \otimes c_j \leq b_j \otimes (a_j \odiv b_j)$ 
	for all $j \leq m$, the latter by monotonicity of $\otimes$.
	If $b_n \otimes c_n < a_n$ for some $n \leq m$
	we are done. 
	Otherwise, if $m = \gamma(a,b) < l$ then 
	$c_{\gamma(a,b)} \in C(A)$ and we are done since
	$b_p \otimes c_p = \bot \leq a_p$ for all $p > {\gamma(a,b)}$.
	Finally, if $m = l \leq \gamma(a,b)$ then 
	$b_l \otimes (a_l \odiv b_l) \leq a_l  = b_l \otimes c_l \leq b_l \otimes (a_l \odiv b_l)$
	since $l \leq \delta(a,b)$, hence either a contradiction
	if $b_l \in I(A)$ or $b_p \otimes c_p = \bot = a_p$ for all $p > l$
	if $b_l \in C(A)$ and consequently $a_l \in C(A)$.
\end{description}
\end{proof}

%

From the propositions above it is straightforward to derive Th.~\ref{prop:lexiRes}, 
which states that, given a residuated POM, it is possible to define a lexicographic 
order on its tuples, which is a residuated POM as well. 

\begin{theorem}\label{prop:lexiRes}
	Let $\langle A, \leq, \monop, \odiv, \1 \rangle$ be a residuated POM with bottom
	element $\bot$.
	Then so is $\langle Lex_k(A), \leq_k, \monop^k, \odiv_k, \1^k \rangle$ for all $k$,
	with $\odiv_k$ defined as 
	\[
	a \odiv_k b = 	\begin{cases}
	                                \begin{array}{lcl}
        		                         (a_1 \odiv b_1) \ldots
	                                   (a_k  \odiv b_k) & & \mbox{if  } k+1 = \gamma(a,b) = \delta(a,b) \\
                                           (a_1 \odiv b_1) \ldots  (a_{\gamma(a, b)}  \odiv b_{\gamma(a, b)}) 	
	                                   \bot^{k - \gamma(a, b)} & & \mbox{if  } k+1 \neq \gamma(a,b) \leq \delta(a,b) \\
           	                         (a_1 \odiv b_1) \ldots  (a_{\delta(a, b)}  \odiv b_{\delta(a, b)})
	                                  (\bigvee Lex_{k-\delta(a,b)}(A)) & & \mbox{otherwise}
	                                 \end{array}
	\end{cases}
	\]
\end{theorem}


\subsection{Infinite tuples}\label{sec:infinite}
We can now move to POMs whose elements are tuples of infinite length.

\begin{proposition}
	Let $\langle A, \leq, \monop, \1 \rangle$ be a POM with bottom element $\bot$.
	Then we can define a POM 
	$\langle Lex_\omega(A), \leq_\omega, \monop^\omega, \1^\omega \rangle$ 
	with bottom element $\bot^\omega$ such that $\monop^\omega$ is defined point-wise
	and
	
	\begin{itemize}
		\item $Lex_\omega(A) = I(A)^\omega \cup I(A)^\ast A \{\bot\}^\omega$
		\item $a \leq_\omega b$ if $a_{\leq k} \leq_k b_{\leq k}$ for all $k$
	\end{itemize}
\end{proposition}

A straightforward adaptation of Prop.~\ref{def:lexilist}.
Thus, we can define a POM of infinite tuples simply by lifting the 
family of POMs of finite tuples.

\begin{remark}
	Note that the seemingly obvious POM structure cannot be 
	lifted to $\bigcup_k Lex_k(A) =  I(A)^\ast A \{\bot\}^\ast$: 
	it would be missing the identity of the monoid.
\end{remark}

\begin{proposition}\label{prop:lexiSLM}
	Let $\langle A, \leq, \monop, \1 \rangle$ be a finitely distributive SLM (distributive CLM).
	Then so is $\langle Lex_\omega(A), \leq_\omega, \monop^\omega, \1^\omega \rangle$.
\end{proposition}

Also a straightforward adaptation, this time of Th.~\ref{theo:lexiSLM}.

\begin{proposition}\label{prop:lexiResOmega}
	Let $\langle A, \leq, \monop, \odiv, \1 \rangle$ be a residuated POM with bottom.
	Then so is $\langle Lex_\omega(A), \leq_\omega, \monop^\omega, \odiv_\omega, \1^\omega \rangle$,
	with $\odiv_w$ defined as
	\[
	a \odiv_\omega b = \begin{cases}
	                                \begin{array}{lcl}
        		                         (a_1 \odiv b_1) \ldots
	                                   (a_k  \odiv b_k) \ldots & & \mbox{if  } \infty = \gamma(a,b) = \delta(a,b) \\
                                           (a_1 \odiv b_1) \ldots  (a_{\gamma(a, b)}  \odiv b_{\gamma(a, b)}) 	
	                                   \bot^\omega & & \mbox{if  } \infty \neq \gamma(a,b) \leq \delta(a,b) \\
           	                         (a_1 \odiv b_1) \ldots  (a_{\delta(a, b)}  \odiv b_{\delta(a, b)})
	                                  (\bigvee Lex_{\omega}(A)) & & \mbox{otherwise}
	                                 \end{array}
		\end{cases}
		\]
\end{proposition}

It follows from Th.~\ref{prop:lexiRes}, via the obvious extension of Lem.~\ref{limit}.
Note that $\bigvee Lex_\omega(A)$ is $\top^\omega$
if $\top \in I(A)$, and  $\top \bot^\omega$ otherwise.



\section{Mini-bucket elimination for residuated POMs}\label{sec:bucket}
This section shows an application of residuation to a general approximation algorithms
for soft CSPs, \emph{Mini-Bucket Elimination} (MBE)~\cite{minibucket},
a relaxation of a well-known complete inference algorithm, \emph{Bucket Elimination} (BE)~\cite{bucket}.

BE first partitions the constraints into \emph{buckets}, where the bucket
of a variable stores those constraints whose \emph{support}\footnote{The \emph{support} of a constraint is the set of variables on which assignment it depends.} contains that variable and none that is higher in the ordering: variables are previously sorted according to some criteria (e.g., just lexicographically on their names: $v_1, v_2, \dots$).  The next step is to process the buckets from top to bottom. When the bucket
of variable $v$ is processed, an \emph{elimination procedure} is performed over the constraints in its bucket, yielding a new constraint defined over all the variables mentioned in the bucket, excluding $v$. This constraint summarises the
``effect'' of $v$ on the remainder of the problem. The new constraint ends up in a lower bucket. BE  finds the preference of the optimal solution and not an approximation of it; however, BE is exponential in the induced width, which measures the aciclicity of a problem.

On the other hand, MBE takes advantage of a control parameter $z$: it partitions the buckets into smaller subsets
called \emph{mini-buckets}, such that their arity is bounded by $z$. Therefore, the cost of computing this approximation is now exponential in $z$, which allows trading off time and space for accuracy. MBE is often used for providing  bounds in branch-and-bound algorithms (see Sect.~\ref{sec:softbb}).

Algorithm~\ref{alg:soft:bucket} extends MBE to work on residuated monoids, hence including also the framework of preferences presented in 
Sect.~\ref{sec:collapsing} and Sect.~\ref{sec:lexico}.  
The algorithm takes as input a problem $P$ defined as $P = \langle V,D,C\rangle_{\mathit{POM}}$, where $V$ is the set of variables $\{v_1, \dots, v_n\}$, 
$D$ is a set of domains $\{D_1, \dots, D_n\}$ (where $v_1 \in D_1, \dots, v_n \in D_n$), $C$ is a set of constraints where 
$\bigcup_{c \in C} \mathit{supp}(c) = V$,\footnote{For instance, a binary constraint $c$ with $\mathit{supp}(c) = \{v_1, v_2\}$ is a function 
$c : (V \longrightarrow D) \longrightarrow A$ that depends only on the assignment of variables $\{v_1, v_2\} \subseteq V$.} and finally, 
the problem is given on a residuated SLM.

We define a \emph{projection} 
operator $\Downarrow$ for a constraint $c$ and variable $v$ as
$(c\Downarrow_v)  = \bigvee_{d \in D_v} c [v:=d].$
Projection decreases the support:
$supp(c\Downarrow_v) \subseteq supp({c}) \setminus \{v\}$. In Algorithm~\ref{alg:soft:bucket} we  use this operator  to eliminate variables from constraints.

At line $4$ Algorithm~\ref{alg:soft:bucket} finds bucket $\mathfrak{B}_i$, which contains all the constraints having $v_i$  in their support. Then at line $5$ we find a partition of $\mathfrak{B}_i$ into $p$ mini-buckets $\mathfrak{Q}$ limited by $z$. All the mini-buckets are projected over $v_i$, thus eliminating it from the support and obtaining a new constraint $g_{i,j}$ as result (line $7$).  Finally, the  bucket $\mathfrak{B}_i$ is discarded from the problem while adding $p$ new constraints $g$ (line $8$). The elimination of the last variable produces an empty-support constraints, whose composition provides the desired upper bound (that is, a solution of $P$ cannot have a better preference than this bound).

Bucket elimination is defined in Algorithm~\ref{alg:bucket}. The second part (from line $7$ to $14$) has been modified with respect to the one in e.g~\cite{bucket} in order to manage partially ordered preferences (as POMs can do).  Note that the $\cdot$ operator extends an assignment tuple $t$ with a new element.
The set $I$ stores all the domain values that produce new undominated tuples, which are saved in $T$. This is repeated for all the assignments in the set of partial solutions, which is finally updated in $\mathit{BSols}$ with not dominated solutions only (line $13$); $g_1$ is the empty-support constraint which represents the (best) preference of such solutions.

By having defined residuation on lexicographic orders, it is now possible to use it in order to have an estimation about how far a partitioning is from buckets: we can  use $\odiv$ to compute good bucket partitions, similarly to the method adopted in \cite{bucketsemiring}.
Let us  consider a partition $\mathfrak{Q} = \{\mathfrak{Q}_1, \mathfrak{Q}_2, \dots, \mathfrak{Q}_p \}$ of a bucket $\mathfrak{B}_i$, which contains all the constraints with variable $v_i$ in the support. We say that $ \mathfrak{Q}$ is a $z$ partition if the support size of its mini-buckets is smaller than $z$, i.e., if $\forall i. |\mathit{supp}(\mathfrak{Q}_i)| \leq z$. The approximation $\mu^\mathfrak{Q}$ 
of the bucket 
is  computed as

\footnotesize
$$\mu^\mathfrak{Q} = \bigotimes\limits^{p}_{j=1} \bigg(\left(\bigotimes\limits \mathfrak{Q}_j\right)\Downarrow_{v_i}\bigg)$$
\normalsize

\begin{algorithm}[t]
	\caption{Mini-Bucket for Residuated POMs.}
	\label{alg:soft:bucket}
	\scriptsize
\hspace*{\algorithmicindent} \textbf{Input:} $P = \langle V,D,C \rangle_{\mathit{POM}}$ and control parameter $z$ \\
\hspace*{\algorithmicindent} \textbf{Output:}  An upper bound of $(\bigotimes_{c \in C}) \Downarrow_V$
	\begin{algorithmic}[1]
		\Function{MBE}{}
		\State $\{v_1, v_2, \dots, v_n\} := \mathit{compute\_order}(P)$
		\For{$i = n \textrm{ to } 1$}
		\State $\mathfrak{B}_i := \{c \in C \mid v_i \in supp(c) \}$
		\State $\{\mathfrak{Q}_1, \mathfrak{Q}_2, \dots, \mathfrak{Q}_p\} := \mathit{partition(\mathfrak{B}_i, z)}$
		\For{$j = 1 \textrm{ to } p$}
			\State $g_{i,j} := (\bigotimes_{c\in \mathfrak{Q}_j} c) \Downarrow_{v_i}$ 
		\EndFor
		\State $C := (C \cup \{g_{i,1}, \dots, g_{i,j}\}) -  \mathfrak{B}_i$
		\EndFor
		\State \Return{$(\bigotimes_{c \in C} c)$}
		\EndFunction
	\end{algorithmic}
	\normalsize
\end{algorithm}

\begin{algorithm}[t]
	\caption{Bucket for Residuated POMs.}
	\label{alg:bucket}
	\scriptsize
	\hspace*{\algorithmicindent} \textbf{Input:} $P = \langle V,D,C \rangle_{\mathit{POM}}$ \\
	\hspace*{\algorithmicindent} \textbf{Output:}  The set of best solutions of $P$
	\begin{algorithmic}[1]
	\Function{BE}{}
	\State $\{v_1, v_2, \dots, v_n\} := \mathit{compute\_order}(P)$
	\For{$i = n \textrm{ to } 1$}
	\State $\mathfrak{B}_i := \{c \in C \mid v_i \in supp(c) \}$
	\State $g_{i} := (\bigotimes_{c\in \mathfrak{B}_j} c) \Downarrow_{v_i}$ 
	\State $C := (C \cup \{g_{i}\}) -  \mathfrak{B}_i$
	\EndFor
	    \State $\mathit{BSols} := \{\langle \rangle\}$ \Comment{The empty tuple}
	\For{$i = 1 \textrm{ to } |V|$}
	\State $T = \emptyset$
	\ForAll {$t \in \mathit{BSols}$}
	\State $I := \{ d \mid \nexists d'.
	(\bigotimes\limits \mathfrak{B}_i) (t \cdot (x_i = d)) < 
	(\bigotimes\limits \mathfrak{B}_i) (t \cdot (x_i = d'))\}$
	\State $T := T \cup \{t \cdot (x_i = d) \mid \exists d \in I\}$ 
	\EndFor
	\State $\mathit{BSols} := T \setminus \{ t \in T \mid \exists t' \in T. (\bigotimes\limits \mathfrak{B}_i) (t) < (\bigotimes\limits \mathfrak{B}_i) (t') \}$
	\EndFor
	\State \Return{$(g_1, \mathit{BSols})$}
	\EndFunction
\end{algorithmic}
	\normalsize
\end{algorithm}

It is noteworthy that residuation may help in quantifying the distance between a bucket 
and its partitioning 

\footnotesize
$$\bigg(\left(\bigotimes\limits \mathfrak{B} \right)\Downarrow_{v_i}\bigg) \smallodiv \left(\bigotimes\limits^{p}_{j=1} \bigg(\left(\bigotimes \mathfrak{Q}_j\right)\Downarrow_{v_i}\bigg)\right)$$
\normalsize

We can compute a refined approximation for a mini-bucket $\mathit{app}_{\mathfrak{Q}_j}$ with respect to the partitioned bucket as

\footnotesize
$$\Bigg(\bigg(\left(\bigotimes\limits \mathfrak{B} \right) \smallodiv \left( \bigotimes\limits (\mathfrak{B} \setminus \mathfrak{Q}_j) \right)\bigg)\Downarrow_{v_i}\Bigg) \smallodiv \bigg(\left(\bigotimes\limits \mathfrak{Q}_j\right)\Downarrow_{v_i}\bigg)$$
\normalsize

If we compose this approximation for each mini-bucket we get an approximation between a bucket and its partitioning
\footnotesize
$$\mathit{approx}_{\mu^\mathfrak{Q}} = \bigotimes\limits_{j} \mathit{approx}_{\mathfrak{Q}_j}$$
\normalsize


\subsection{Soft Branch-and-Bound}\label{sec:softbb}
Algorithms as MBE can be used to obtain a lower bound that underestimates the best solution of a given  problem $P = \langle V,D,C \rangle_{\mathit{POM}}$. 
This bound can be then passed as input  to a search algorithm in order to increase its pruning efficiency~\cite{jheuristics}. In the following of this section, we describe an 
example of search that can be used to find all the solutions of a (possibly lexicographic) soft CSP. Note that this algorithm is designed to deal with partially ordered solutions. 
On the contrary, in  \cite{schiex} the solution of  \emph{Lex-VCSP}  (and in general \emph{Valued CSP}s, i.e., \emph{VCSP}s) is associated with a set of totally ordered preferences.

The family of \emph{Soft Branch-and-Bound} algorithms explores the state space of a soft CSP as a tree. A \emph{Depth-First Branch-and-Bound} (\emph{DFBB})  (see Algorithm~\ref{alg:soft:bnb}) performs a depth-first traversal of the search tree. Given a partial assignment $t$ of $V$, an upper bound $\mathit{ub}(t)$ is an overestimation of the acceptance degree of any possible complete assignment involving $t$. A lower bound $\mathit{lb}(t)$ is instead a minimum acceptance degree that we are willing to accept during the search.

With each node in the search tree is associated a set of variables $X \subseteq V$ that have been already assigned (and the set of unassigned ones is 
given by $U = V \setminus X$), 
along with the associated (partial) assignment $t$ to those variables ($\mathit{supp}(t) = X$). A leaf node is associated with a complete assignment ($\mathit{supp}(t) = V$).
Each time a new internal node is created, a variable $v_i \in U$ to assign next is chosen, as well as an element $d \in D_{v_i}$ of its domain. Note that the procedure in 
 Algorithm~\ref{alg:soft:bnb} prunes the search space at line $8$, since it only explores those assignments $v_i = d$ such that there exists an upper bound $u \in \mathit{UB}$ 
that is better than a lower bound $l \in \mathit{LB}$.


The efficiency of Soft DFBB depends largely on its pruning capacity, which relies on the quality of its bounds: the higher $\mathit{lb}$ and the lower $\mathit{ub}$ (still ensuring they are actual bounds of optimal solutions), the better Soft DFBB performs. Note that, in order to deal with partial orderings of preferences, Algorithm~\ref{alg:soft:bnb} has to manage sets of undominated upper \emph{UB} and lower \emph{LB} bounds of (partial) solutions, differently from classical Branch-and-Bound. In Algorithm~\ref{alg:soft:bnb},  \emph{LB} returns a set of lower bounds for a given partial assignment. When all the variables are assigned (line $2$), the procedure stops with a solution.

\begin{Algorithm}[t]{11cm}
  \caption{Soft Depth-First Branch-and-Bound.}
    \label{alg:soft:bnb}
    \scriptsize
  \begin{algorithmic}[1]
    \Function{SoftDFBB}{$t, \mathit{LB}$}
   \If {$(\mathit{supp}(t) = V)$}
     \State \Return{$\bigotimes C(t)$}
   \Else  
   \State \textbf{let} $v_i \in U$ \Comment $U$ is the set of unassigned variables
    \ForAll{$d \in D_{v_i}$}
      \State $H:= \mathit{UB}(t \cdot (v_i = d))$
      \If {$(\exists u \in H, \exists l \in LB. \: l \leq_{\mathit{POM}} u)$}
        \State \ $\mathit{LB} :=  \mathit{LB} \cup \mathtt{SoftDFBB}(t \cdot (v_i = d), \mathit{LB})$
        \State $\mathit{LB} := \mathit{LB} \setminus \{e \in \mathit{LB} \mid \exists e' \in \mathit{LB}. e <_{\mathit{POM}}  e'\}$
      \EndIf
    \EndFor
    \State \Return{$LB$}
    \EndIf
    \EndFunction
  \end{algorithmic}
\end{Algorithm}

%
%

\section{Conclusions and Future Works}\label{sec:conclusion}
In this paper we considered a formal framework for soft CSP based on a residuated monoid of partially ordered preferences. 
This allows for using the classical solving algorithms that need some preference removal, as for instance arc consistency where values need to be moved 
from binary to unary constraints, or for proving a cost estimation to be used during the search for solutions, as for instance branch-and-bound algorithms.
%
The contribution of this paper is twofold. On the one side, we proved the adequacy of the formalism for modelling lexicographic orders. 
On the other side, we showed how it can enable heuristics for efficiently solving soft CSPs, such as the Nucket and Mini-bucket elimination.
%

Our focus on soft CSP includes its computational counterparts based on constraints, such as soft CCP~\cite{labelled}, and in fact, considering infinite tuples 
enables to model temporal reasoning as shown for soft constraint automata in \cite{sca}. However, the framework is reminiscent, and is in fact an extension, 
of previous formalisms such as monotonic logic programming~\cite{residuatedlogic}, whose semantics is given in terms of residuated 
lattices of complete lattices of truth-values. And it fits in the current interests on the development of sequent systems for substructural logics, 
as witnessed by current research projects~\cite{ciabattoni}:
%
well-known examples are Lukasiewicz's many-valued logics, relevance logics and linear logics.

All the connections sketched above deserve further investigations. For the time being, we leave to future work some related extensions.
Mini-bucket is often used for providing an upper bound in branch-and-bound algorithms: for this reason we will investigate this technique, 
as well as other solving methods used in the solution of lexicographic problems~\cite{freuderlexi}. 
We will also study ad-hoc heuristics for selecting the order in Algorithm~\ref{alg:soft:bucket}, directly depending on lexicographic orders.

\bibliographystyle{splncs03}
\bibliography{main}
\end{document}